\documentclass[a4paper,fleqn]{article}
\usepackage{fullpage}

\usepackage[utf8]{inputenc}

\usepackage{amsmath}
\usepackage{amssymb}
\usepackage{amsthm}

\newtheorem{lemma}{Lemma}
\newtheorem{theorem}{Theorem}
\newtheorem{corollary}{Corollary}
\newtheorem{example}{Example}

\usepackage{mathtools}
\usepackage{xspace}
\usepackage{tikz}
\usetikzlibrary {arrows.meta,automata,positioning}

\newcommand{\myquot}[1]{``#1''}

\renewcommand{\epsilon}{\varepsilon}
\newcommand{\set}[1]{\{#1\}}
\newcommand{\pow}[1]{2^{#1}}
\newcommand{\nats}{\mathbb{N}}
\newcommand{\size}[1]{|#1|}

\newcommand{\dtmcsymbol}{\mathcal{M}}
\newcommand{\dtmc}{DTMC\xspace}
\newcommand{\ap}{AP}
\newcommand{\initmark}{I}
\newcommand{\path}{\pi}
\newcommand{\pointedpaths}[2]{\Pi(#1,#2)}
\newcommand{\allpaths}[1]{\Pi(#1)}
\newcommand{\pathprefix}{\rho}
\newcommand{\state}[2]{#1(#2)}

\newcommand{\suffix}[2]{#1[#2,\infty)}
\newcommand{\probmes}{\mu}
\newcommand{\cyl}[1]{C_{#1}}

\newcommand{\ltl}{LTL\xspace}
\newcommand{\ctl}{CTL\xspace}
\newcommand{\rctl}{rCTL\xspace}
\newcommand{\rltl}{rLTL\xspace}
\newcommand{\pctl}{PCTL\xspace}
\newcommand{\rpctl}{rPCTL\xspace}
\newcommand{\ctlstar}{CTL*\xspace}
\newcommand{\atlstar}{ATL*\xspace}
\newcommand{\atl}{ATL\xspace}
\newcommand{\pctlstar}{PCTL*\xspace}
\newcommand{\rpctlstar}{rPCTL*\xspace}

\newcommand{\prob}[2]{\mathcal{P}_{#1}(#2)}
\newcommand{\probsymb}{\mathcal{P}}
\renewcommand{\phi}{\varphi}
\newcommand{\cl}{\mathrm{cl}}
\newcommand{\reval}{V}
\newcommand{\boolfour}{\mathbb{B}_4}

\newcommand{\pspace}{\textsc{PSpace}\xspace}
\newcommand{\ptime}{\textsc{PTime}\xspace}

\tikzset{robust/.style={line width=.16ex,line join=round}}

\let\Box\relax
\DeclareMathOperator{\Box}{%
	\text{%
		\tikz[baseline]{%
    			\draw[robust] (0ex,-.1ex) -- (0ex, 1.4ex) -- (1.5ex, 1.4ex) -- (1.5ex, -.1ex) -- cycle;%
		}%
	}%
}

\DeclareMathOperator{\Boxdot}{%
	\text{%
		\tikz[baseline]{%
    			\draw[robust] (0ex, -.1ex) -- (0ex, 1.4ex) -- (1.5ex, 1.4ex) -- (1.5ex, -.1ex) -- cycle;%
	    		\fill (.75ex, .65ex) circle (.15ex);%
    		}%
	}%
}

\let\Diamond\relax
\DeclareMathOperator{\Diamond}{%
	\text{%
		\tikz[baseline]{%
			\draw[robust] (0ex,.6ex) -- (.95ex, 1.55ex) -- (1.9ex, .6ex) -- (.95ex, -.35ex) -- cycle;%
		}%
	}%
}

\DeclareMathOperator{\Diamonddot}{%
	\text{%
		\tikz[baseline]{%
			\draw[robust] (0ex,.6ex) -- (.95ex, 1.55ex) -- (1.9ex, .6ex) -- (.95ex, -.35ex) -- cycle;%
			\fill (.95ex, .6ex) circle (.15ex);%
		}%
	}%
}

\DeclareMathOperator{\Next}{%
	\text{%
		\tikz[baseline]{%
    			\draw[robust] (.75ex, .65ex) circle (.75ex);%
    		}%
	}%
}

\DeclareMathOperator{\Nextdot}{%
	\text{%
		\tikz[baseline]{%
    			\draw[robust] (.75ex, .65ex) circle (.75ex);%
	    		\fill (.75ex, .65ex) circle (.15ex);%
    		}%
	}%
}

\DeclareMathOperator{\Until}{%
	\text{%
		\tikz[baseline]{%
			\node[inner sep=0pt, anchor=base, font=\bfseries] {U};%
		}%
	}%
}

\DeclareMathOperator{\Untildot}{%
	\text{%
		\tikz[baseline]{%
			\node[inner sep=0pt, anchor=base, font=\bfseries] {U};%
			\fill (.1ex, .9ex) circle (.15ex);%
		}%
	}%
}

\DeclareMathOperator{\Releasedot}{%
	\text{%
		\tikz[baseline]{%
			\node[inner sep=0pt, anchor=base, font=\bfseries] {R};%
			\fill (-.025ex, 1.175ex) circle (.15ex);%
		}%
	}%
}

\newcommand{\sat}[2]{\mathrm{Sat}(#1,#2)}

\newcommand{\safety}{\mathrm{Safe}}
\newcommand{\cobuchi}{\mathrm{CoB\ddot{u}chi}}
\newcommand{\buchi}{\mathrm{B\ddot{u}chi}}
\newcommand{\reach}{\mathrm{Reach}}
\newcommand{\untillang}{\mathrm{Until}}

\begin{document}




\title{Robust Probabilistic Temporal Logics}
\author{Martin Zimmermann\\ \small Aalborg University, Aalborg, Denmark.}
\date{\vspace{-.8cm}}
\maketitle
\begin{abstract}
We robustify PCTL and PCTL$^*$, the most important specification languages for probabilistic systems, and show that robustness does not increase the complexity of their model-checking problems.
\end{abstract}

\section{Introduction}

Specifications of reactive systems are typically implications~$\phi_a \rightarrow \phi_g$ where $\phi_a$ is an environment assumption and $\phi_g$ is a system guarantee, both specified in a temporal logic.  
Such a specification is satisfied whenever the assumption is violated, independently of the system's behaviour.
Assume, for example, that both the assumption and the guarantee are  invariants~$\phi_a = \Box \psi_a$ and $\phi_g = \Box \psi_g$ for propositional formulas~$\psi_g$ and $\psi_a$.
Then, the specification~$\Box \psi_a \rightarrow \Box \psi_g$ is satisfied if the formula~$\psi_a$ is violated just once, even if the formula~$\psi_g$ never holds.
Such a behaviour is clearly undesirable, but the classical semantics of temporal logics are not sufficiently robust to deal with violations of the environment assumption.

Considerable effort has been put into overcoming this \myquot{defect} to provide robust semantics for temporal logics. 
However, the notion of robustness is hard to formalize, which is witnessed by the plethora of incomparable notions of robustness in the literature on verification (see, e.g., the introduction of~\cite{DBLP:journals/tocl/AnevlavisPNT22} for a recent overview). 
Here, we further develop an approach due to Tabuada and Neider based on a novel, robust semantics for temporal logics, originally introduced for Linear Temporal Logic (\ltl)~\cite{TabuadaNeider}.
They argue that there are four canonical degrees a formula of the form~$\Box \psi$ can be violated:
\begin{enumerate}
    \item $\psi$ is violated only finitely often.
    \item $\psi$ is violated infinitely often, but also holds infinitely often.
    \item $\psi$ is satisfied only finitely often.
    \item $\psi$ is never satisfied.
\end{enumerate}
Note that there is a natural order between these cases. 
Consequently, their robust semantics uses five truth values, one for satisfaction and four more to capture the four degrees of violation.
Furthermore, Tabuada and Neider defined the semantics of implication such that $\Box \psi_a \rightarrow \Box \psi_g$ is satisfied whenever the degree of violation of the guarantee~$\Box \psi_g$ is not more severe than the violation of the assumption~$\Box \psi_g$.
Thus, the semantics indeed robustly handles violations of environment assumptions.

The resulting logic, called robust \ltl (\rltl), has been extensively studied with very encouraging results: robustness can be added without increasing the complexity of model-checking and synthesis~\cite{DBLP:journals/tocl/AnevlavisPNT22,TabuadaNeider,rltlgames}, robust semantics increases the usefulness of runtime monitoring~\cite{runtime}, and \rltl can even be extended with increased expressiveness or timing constraints, again without an increase in complexity~\cite{DBLP:journals/iandc/NeiderWZ22}.
This approach towards robustness even extends to other temporal logics, e.g., branching-time logics like \ctl and \ctlstar~\cite{rctljournal} and alternating-time logics like \atl and \atlstar~\cite{MNZ23}, where robustness can again be added without increasing the complexity of the most important verification problems.

Beyond the fact that this form of robustness comes for free (in terms of computational complexity), it only changes the semantics of the logics, but not the syntax. 
Furthermore, these logics are also evaluated over classical transition systems with the classical binary satisfaction relation for atomic propositions, i.e., robustness does not emerge from multi-valued semantics of the models (which might be hard to determine), but purely from the semantics. 
These aspects allow for a smooth transition from classical semantics to robust semantics for temporal logics. 
In conclusion, Tabuada and Neider introduced a natural and lightweight approach to add robustness that is applicable to a wide range of logics.

However, these logics capture only robustness in the temporal dimension, i.e., they are concerned with a single execution.
Statements like \myquot{99\% of the executions answer each request eventually} require robustness in terms of the whole set of executions, which is orthogonal to the capabilities of the robust logics studied thus far.
To express such specifications, Hansson and Jonsson introduced probabilistic \ctl (\pctl)~\cite{pctl}, while Aziz, Singhal, and Balarin introduced probabilistic \ctlstar (\pctlstar)~\cite{pctlstar}.
\pctl and \pctlstar replace the existential and universal quantification over paths in \ctl and \ctlstar by the probabilistic operator~$\prob{I}{\Phi}$, where $I \subseteq [0,1]$ is an interval with rational endpoints and $\Phi$ is a property of paths.
Intuitively, $\prob{I}{\Phi}$ is satisfied in a state~$s$ if the probability that a path starting in $s$ satisfies~$\Phi$ is in the interval~$I$.
As \ctl, \pctl requires each temporal operator to be preceded by a $\probsymb$ while \pctlstar (as \ctlstar) allows arbitrary nesting of Boolean connectives, temporal operators, and $\probsymb$.
For example, the property \myquot{99\% of the executions answer each request eventually} is expressed by the \pctlstar formula~$\prob{\ge.99}{\Box (q \rightarrow \Diamond p)}$, where $q$ represents a request and $p$ a response.

In this work, we further the study of robust semantics for temporal logics a la Tabuada and Neider by robustifying \pctl and \pctlstar, obtaining the logics \rpctl and \rpctlstar. 
In line with the design goals of the approach, the robust variants have (essentially) the same syntax as the non-robust variants and are evaluated over the same structures, simplifying the transition from the non-robust to the robust setting. 
The semantics of \rpctl and \rpctlstar also follow the blueprint, i.e., they are five-valued employing the four degrees of violation described above. 
This simplifies the transition from robust semantics for linear, branching, and alternating time to the robust probabilistic setting. 

As our main contribution, we show that this robustification comes again for free: the automata-based model-checking algorithms for \rpctl and \rpctlstar can be generalized to the robust semantics. 
This result is in line with those on the robust temporal logics studied thus far, once more showing the versatility of robustness a la Tabuada and Neider.

\section{Preliminaries}

We denote the set of non-negative integers by $\nats$.
Throughout the paper, we fix a finite set~$\ap$ of atomic propositions we use to label our models and to build our formulas. 
For algorithmic purposes, we assume that all probabilities used in the following are rational.


A discrete-time Markov chain (\dtmc)~$\dtmcsymbol = (S, s_\initmark, \delta, \ell)$ consists of a finite set~$S$ of states containing the initial state~$s_\initmark$, a (stochastic) transition function~$\delta \colon S\times S \rightarrow [0,1]$ satisfying $\sum_{s' \in S}\delta(s,s') = 1$ for all $s \in S$, and a labeling function~$\ell\colon S \rightarrow \pow{\ap}$.
The size~$\size{\dtmcsymbol}$ of $\dtmcsymbol$ is defined as $\sum_{s,s' \in S} \size{\delta(s,s')}$, where $\size{p}$ denotes the length of the binary encoding of $p \in \mathbb{Q}$.

A path of $\dtmcsymbol$ is an infinite sequence~$\path = s_0 s_1 s_2 \cdots \in S^\omega$ such that $\delta(s_n, s_{n+1}) >0$ for all $n \in \nats$.
We say that $\path$ starts in $s_0$. For $n \in \nats$,  we write $\state{\path}{n} = s_n$ for the $n$-th state of $\path$ 
and $\suffix{\path}{n} = s_n s_{n+1}s_{n+2}\cdots$ for the suffix of $\path$ starting at position~$n$. 
We write $\pointedpaths{\dtmcsymbol}{s}$ for the set of all paths of $\dtmcsymbol$ starting in $s \in S$ and define $\allpaths{\dtmcsymbol} = \bigcup_{s \in S} \pointedpaths{\dtmcsymbol}{s}$.

The probability measure~$\probmes_s$ on sets of paths starting in some state~$s \in S$ is defined as usual:
Fix some non-empty path prefix~$\pathprefix =s_0 \cdots s_{n}$. The probability of the cylinder set
$\cyl{\pathprefix} = \set{\path \in \pointedpaths{\dtmcsymbol}{s} \mid \pathprefix \text{ is a prefix of }\path}$
is
\[
\probmes_{s}(\cyl{\pathprefix}) = 
\begin{cases}
    0 &\text{if $s_0 \neq s$,}\\
    \prod\nolimits_{j = 0}^{n-1} \delta(\state{\pathprefix}{j}, \state{\pathprefix}{j+1}) &\text{if $s_0 = s$.} 
\end{cases}
\]
Using Carathéodory's extension theorem, we lift $\probmes_{s}$ to a measure on the $\sigma$-algebra induced by the cylinder sets of path prefixes starting in $s$ (see, e.g.,~\cite[Theorem 1.41]{klenke} for details.
All sets of paths used in the following are $\omega$-regular (see, e.g., \cite[Chapter 1]{gtw} for background on $
\omega$-regular languages) and therefore measurable.

\section{Robust PCTL}
In this section, we robustify \pctl~\cite{pctl}. Following the general design goals of the robustification a la Tabuada and Neider, robust \pctl (\rpctl) and \pctl share the same syntax (but for the dots to distinguish them), i.e., the formulas of \rpctl are given by the grammar
\begin{align*}
\phi  \Coloneqq {}&{} p \mid \neg \phi \mid \phi \wedge \phi \mid \phi \vee \phi \mid \phi \rightarrow \phi  \mid \\
 {}&{} \prob{\sim\lambda}{\Nextdot\phi}\mid\prob{\sim\lambda}{\Diamonddot \phi}\mid \prob{\sim\lambda}{\Boxdot \phi} \mid \prob{\sim\lambda}{\phi \Untildot \phi} \mid \prob{\sim\lambda}{\phi \Releasedot \phi}    
\end{align*}
where $p$ ranges over $\ap$, ${\sim} \in \set{<,\le, =, \ge, >}$, and $\lambda \in [0,1]$ is a rational probability threshold.
The size~$\size{\phi}$ of a formula~$\phi$ is defined as the number of subformulas of $\phi$ plus the maximal length~$\size{\lambda}$ of the binary encodings of the thresholds~$\lambda \in \mathbb{Q}$ appearing in $\phi$.
Note that we have all Boolean operators in the grammar, as the semantics of negation is non-classical and implication cannot be derived from negation and disjunction. This is due to the five-valued robust semantics (see \cite[Section 3.3]{TabuadaNeider} for a detailed discussion).
For didactic reasons, we also prefer to explicitly have the operators eventually~($\Diamonddot$) and always~($\Boxdot$) as they already capture the essence of the robust semantics. The robust semantics of until~($\Untildot$) and release~($\Releasedot$) then generalize these.

Again, following the design goals of the robustification a la Tabuada and Neider, \rpctl is evaluated over the same structures as \pctl, i.e., over discrete-time Markov chains. Let $\dtmcsymbol = (S, s_\initmark, \delta, \ell)$ be a \dtmc. The semantics of \rpctl is defined via an evaluation function~$\reval_\dtmcsymbol$ mapping a vertex~$s$ of $\dtmcsymbol$ and a formula~$\phi$ to a truth value in the (ordered) set~$\boolfour = \set{1111\succ 0111\succ 0011\succ 0001\succ 0000}$.
Given a truth value~$t = b_1 b_2 b_3 b_4 \in\boolfour$, we write $t[k]$ for $b_k$.

The evaluation function is defined inductively via
\begin{itemize}
    
    \item $\reval_\dtmcsymbol(s, p) = \begin{cases} 
    1111 & \text{if } p \in \ell(s),\\
    0000 & \text{if } p \notin \ell(s),
    \end{cases}
    $
    
    \item $\reval_\dtmcsymbol(s, \neg \phi) = \begin{cases} 
    1111 & \text{if } \reval_\dtmcsymbol(s, \phi) \prec 1111,\\
    0000 & \text{if } \reval_\dtmcsymbol(s,\phi) = 1111,
    \end{cases}
    $
    
    \item $\reval_\dtmcsymbol(s, \phi_0 \wedge \phi_1) = \min(\reval_\dtmcsymbol(s, \phi_0), \reval_\dtmcsymbol(s, \phi_1))$,
    
    \item $\reval_\dtmcsymbol(s, \phi_0 \vee \phi_1) = \max(\reval_\dtmcsymbol(s, \phi_0), \reval_\dtmcsymbol(s, \phi_1))$,
    
    \item $\reval_\dtmcsymbol(s, \phi_0 \rightarrow \phi_1) = \begin{cases} 
    1111 & \text{if } \reval_\dtmcsymbol(s, \phi_0) \preceq \reval_\dtmcsymbol(s, \phi_1),\\
    \reval_\dtmcsymbol(s, \phi_1) & \text{if } \reval_\dtmcsymbol(s, \phi_0) \succ \reval_\dtmcsymbol(s, \phi_1),
    \end{cases}
    $
    
    \item $\reval_\dtmcsymbol(s, \prob{\sim\lambda}{\Nextdot\phi}) = b_1b_2b_3b_4 \in \boolfour$ where for all $k \in \set{1,2,3,4}$: $b_k = 1$ iff $\probmes_s(\set{\path \in \pointedpaths{\dtmcsymbol}{s} \mid \reval_\dtmcsymbol(\state{\path}{1},\phi)[k] = 1})\sim\lambda$,
    
    \item $\reval_\dtmcsymbol(s, \prob{\sim \lambda}{\Diamonddot \phi}) = b_1 b_2 b_2 b_4\in \boolfour$ where for all $k \in \set{1,2,3,4}$: $b_k = 1$ iff $\probmes_s(\set{\path \in \pointedpaths{\dtmcsymbol}{s} \mid (\reval_\dtmcsymbol(\state{\path}{n},\phi))[k] = 1  \text{ for some } n\in\nats})\sim \lambda$, and 
    
    \item $\reval_\dtmcsymbol(s, \prob{\sim \lambda}{\Boxdot \phi}) = b_1b_2b_3b_4\in \boolfour$ with 
    \begin{itemize}
        
        \item $b_1 = 1$ iff $\probmes_s(\set{\path \in \pointedpaths{\dtmcsymbol}{s} \mid (\reval_\dtmcsymbol(\state{\path}{n},\phi))[1] = 1  \text{ for all } n\in\nats})\sim \lambda$,
        
        \item $b_2 = 1$ iff $\probmes_s(\set{\path \in \pointedpaths{\dtmcsymbol}{s} \mid (\reval_\dtmcsymbol(\state{\path}{n},\phi))[2] = 1  \text{ for all but finitely many } n\in\nats})\sim \lambda$,
        
        \item $b_3 = 1$ iff $\probmes_s(\set{\path \in \pointedpaths{\dtmcsymbol}{s} \mid (\reval_\dtmcsymbol(\state{\path}{n},\phi))[3] = 1  \text{ for infinitely many } n\in\nats})\sim \lambda$,
        
        \item $b_4 = 1$ iff $\probmes_s(\set{\path \in \pointedpaths{\dtmcsymbol}{s} \mid (\reval_\dtmcsymbol(\state{\path}{n},\phi))[4] = 1  \text{ for some } n\in\nats})\sim \lambda$,

\end{itemize}

    \item $\reval_\dtmcsymbol(s, \prob{\sim \lambda}{\phi \Untildot \psi}) = b_1 b_2 b_2 b_4\in \boolfour$ where for all $k \in \set{1,2,3,4}$: $b_k = 1$ iff $\probmes_s(\set{\path \in \pointedpaths{\dtmcsymbol}{s} \mid \text{ there exists } n\in\nats \text{ s.t. } (\reval_\dtmcsymbol(\state{\path}{n},\psi))[k] = 1 \text{ and } (\reval_\dtmcsymbol(\state{\path}{n'},\phi))[k] = 1  \text{ for all } n' < n }) \sim \lambda$, and 
    
    \item $\reval_\dtmcsymbol(s, \prob{\sim \lambda}{\phi\Releasedot \psi}) = b_1b_2b_3b_4 \in\boolfour$ with               
    \begin{itemize}
        
        \item $b_1 = 1$ iff $\probmes_s(\set{\path \in \pointedpaths{\dtmcsymbol}{s} \mid \text{for all } n \in \nats~ (\reval_\dtmcsymbol(\state{\path}{n},\psi))[1] = 1  \text{ or } (\reval_\dtmcsymbol(\state{\path}{n},\phi))[1] = 1 \text{ some } n' < n})\sim \lambda$,
        
        \item $b_2 = 1$ iff $\probmes_s(\set{\path \in \pointedpaths{\dtmcsymbol}{s} \mid (\reval_\dtmcsymbol(\state{\path}{n},\psi))[2] = 1  \text{ for all but finitely many } n\in\nats \text{ or } (\reval_\dtmcsymbol(\state{\path}{n},\phi))[2] = 1 \text{ some } n\in \nats })\sim \lambda$,
        
        \item $b_3 = 1$ iff $\probmes_s(\set{\path \in \pointedpaths{\dtmcsymbol}{s} \mid (\reval_\dtmcsymbol(\state{\path}{n},\psi))[3] = 1  \text{ for infinitely many } n\in\nats \text{ or } (\reval_\dtmcsymbol(\state{\path}{n},\phi))[3] = 1 \text{ some } n' < n})\sim \lambda$, and
        
        \item $b_4 = 1$ iff $\probmes_s(\set{\path \in \pointedpaths{\dtmcsymbol}{s} \mid (\reval_\dtmcsymbol(\state{\path}{n},\psi))[4] = 1  \text{ for some } n\in\nats \text{ or } (\reval_\dtmcsymbol(\state{\path}{n},\phi))[4] = 1 \text{ some } n' < n})\sim \lambda$.
        
    \end{itemize}

\end{itemize}

Here, the cases for Boolean connectives and temporal operators follow the blueprint introduced by Tabuada and Neider for \rltl (which also have been used for the robust variants of \ctl, \ctlstar, \atl, and \atlstar). For a detailed motivation and description, we refer to~\cite[Section 3]{TabuadaNeider}. On the other hand, the semantics of $\probsymb$ generalizes the classical two-valued semantics of \pctl to five truth values, just as the path quantifiers in robust \ctl~\cite{rctljournal} generalizes the path quantifiers of \ctl and the strategy quantifier of robust \atl~\cite{MNZ23} generalizes the strategy quantifier of \atl.

\begin{example}
\label{example:rpctl}
Consider the formula~$\phi= \prob{\ge .9}{\Boxdot a} \rightarrow \prob{\ge .95}{\Boxdot g}$ expressing a robust assume-guarantee property. Assume $\phi$ evaluates to $1111$ and consider the following cases:
\begin{itemize}
    \item Assume $\prob{\ge .9}{\Boxdot a}$ evaluates to $1111$, i.e., with probability $\ge .9$, $a$ holds at every position of a path. Then, by the semantics of the implication, with probability $\ge .95$, $g$ holds at every position. 
    \item Assume $\prob{\ge .9}{\Boxdot a}$ evaluates to $0111$, i.e., with probability $\ge .9$, $a$ holds at all but finitely many positions of a path (but not at \emph{every} position of a path with probability $\ge .9$). Then, by the semantics of the implication, with probability $\ge .95$, $g$ holds at least at all but finitely many positions.
    \item Similar arguments hold for the truth values~$0011$ and $0001$: Assume with probability~$\ge .9$, $a$ holds infinitely often ($a$ holds at least once). Then, with probability~$\ge .95$, $g$ holds infinitely often) ($g$ holds at least once).
\end{itemize}
Thus, the semantics of $\phi$ ensures that a violation of the assumption $\Boxdot a$ is met with (at most) a proportional violation of the guarantee~$\Boxdot g$.

But we can even derive useful information if $\phi$ does not evaluate to $1111$.
Assume, $\phi$ evaluates to $t \prec 1111$. 
This can only be the case if the assumption~$\prob{\ge .9}{\Boxdot a}$ evaluates to some truth value strictly smaller than $t$ and the guarantee~$\prob{\ge .95}{\Boxdot g}$ evaluates to $t$.
Hence, even if the implication does not hold, it still yields the degree of satisfaction of the guarantee.
\end{example}

The above example show that the robust semantics does indeed capture the intuition described in the introduction. 

\subsection{Expressiveness}
\label{subsec_exp}

In this section, we discuss the expressiveness of \rpctl; in particular, we compare it to the expressiveness of \pctl.

Our first result shows that \rpctl is at least as expressive as \pctl. It follows directly from the design goals of the robust semantics: they are defined such the first bit represents standard (non-robust) semantics. 

Note that the restriction to implication-free formulas is just technical, as implications~$\phi\rightarrow \psi$ in \pctl formulas can always be rewritten as $\neg\phi \vee \psi$.
The need for the implication-removal stems from the fact that robust implication does not generalize classical implication~\cite[Footnote~3]{runtime}.

\begin{lemma}
\label{lemma:pctl2rpctl}
Let $\phi$ be a \pctl formula without implications, and let $\dtmcsymbol$ be a \dtmc with initial state~$s_\initmark$. Then, $\dtmcsymbol,s_\initmark \models \phi$ iff $\reval_\dtmcsymbol(s_\initmark, \dot{\phi}) = 1111$, where $\dot{\phi}$ is the \rpctl formula obtained from $\phi$ by dotting all temporal operators. 
\end{lemma}

\begin{proof}
By induction over the construction of $\phi$, formalizing the fact that the first bit of the robust semantics captures the classical semantics of \pctl. 
This can be seen by a careful inspection of the robust semantics.
\end{proof}

\begin{corollary}
\rpctl is at least as expressive as \pctl.    
\end{corollary}

Let us briefly discuss the other inclusion, e.g., is \rpctl strictly more expressive than \pctl?
This is true for the non-probabilistic setting,  where \rctl (robust \ctl) is strictly more expressive than \ctl~\cite{rctljournal}, as $\reval_\dtmcsymbol(s_\initmark,\forall\Boxdot{p}) \succeq 0111$ holds iff $p$ holds at all but finitely many positions of every path starting in $s$. 
This property cannot be expressed in \ctl~\cite[Theorem 6.21]{bk}. 
However, the analogous property \myquot{$p$ holds at all but finitely many positions often with probability one} can be expressed in \pctl~\cite[Theorem 10.48]{bk} (when considering finite {\dtmc}s), relying on the fact that a path ends up  with probability one in a bottom strongly-connected component.
We leave open the question whether similar arguments are sufficient to show that \rpctl can be embedded into \pctl (w.r.t.\ finite {\dtmc}s).

Let us conclude this section with a consequence of the embedding proven in Lemma~\ref{lemma:pctl2rpctl}. \rpctl satisfiability asks, given a formula~$\phi$ and a truth value~$t^*$ whether there is a \dtmc~$\dtmcsymbol$ with initial state~$s_\initmark$ such that~$\reval_\dtmcsymbol(s_\initmark, \phi) \succeq t^*$.
\pctl satisfiability has recently been shown to be undecidable~\cite{pctlsatundec}.
So, due to Lemma~\ref{lemma:pctl2rpctl}, which allows us to embed \pctl in \rpctl, \rpctl satisfiability is also undecidable.

\subsection{Model-checking}

In this section, we prove that model-checking \rpctl is not harder than model-checking \pctl, which is in \ptime~\cite{pctl}, i.e., robustness can be added for free.
Formally, \rpctl model-checking is the following problem: Given a \dtmc~$\dtmcsymbol$ with initial state~$s_\initmark$, an \rpctl formula~$\phi$, and a truth value~$t^* \in \boolfour$, is $\reval_\dtmcsymbol(s_\initmark, \phi) \succeq t^*$?

In the following, we prove that \rpctl model-checking is not harder than \pctl model-checking by combining techniques developed for robustified temporal logics with a generalization of an automata-based model-checking algorithm for \pctl. 

\begin{theorem}
\label{theorem:rpctlmodelchecking}
\rpctl model-checking is in \ptime.
\end{theorem}

\begin{proof}
Fix a \dtmc~$\dtmcsymbol = (S, s_\initmark, \delta, \ell)$ and an \rpctl formula~$\phi$, and let $\cl(\phi)$ denote the set of subformulas of $\phi$ (which is defined as expected).
We show how to inductively compute the satisfaction sets
\[\sat{\psi}{t} = \set{s \in S \mid \reval_\dtmcsymbol(s, \psi) \succeq t}\]
for $\psi \in \cl(\phi)$ and $t \in \boolfour$.
Note that $\sat{\psi}{0000} = S$ holds for all subformulas~$\psi$.
Hence, in the following, we only consider $t \succ 0000$.
Also, the cases for atomic propositions and Boolean connectives are trivial, as they amount to Boolean combinations of already computed sets (see, e.g.,~\cite{rctljournal}). 
For example, we have 
$\sat{\psi' \wedge \psi''}{t} = \sat{\psi'}{t} \cap \sat{\psi''}{t}$ and
$\sat{\psi' \vee \psi''}{t} = \sat{\psi'}{t} \cup \sat{\psi''}{t}$.
Hence, it only remains to consider subformulas~$\psi$ of the form~$\prob{\sim \lambda}{\Nextdot\psi'}$, $\prob{\sim \lambda}{\Diamonddot\psi'}$, $\prob{\sim \lambda}{\Boxdot\psi'}$, $\prob{\sim \lambda}{\psi'\Untildot\psi''}$, or $\prob{\sim \lambda}{\psi'\Releasedot\psi''}$.

We begin with the next operator.
Here, we have $s \in \sat{\prob{\sim \lambda}{\Nextdot\psi'}}{t}$ iff \[
\probmes_s(\set{\path \in \pointedpaths{\dtmcsymbol}{s} \mid \state{\path}{1} \in \sat{\psi'}{t}}) = \left(\sum\nolimits_{s' \in \sat{\psi'}{t}} \delta(s,s')\right) \sim \lambda.
\]
The value~$\sum_{s'} \delta(s,s')$ can be computed and compared to $\lambda$ in polynomial time, as $\sat{\psi}{t}$ has already been computed by induction hypothesis.

For the remaining temporal operators, we rely on standard automata-theoretic characterizations of the sets of paths satisfying a temporal formula (see, e.g., \cite[Section 1]{gtw} for an introduction to automata on infinite words).
We will then apply the following result due to Baier et al.: Given a \dtmc~$\dtmcsymbol$, one of its states~$s$, and an unambiguous Büchi automaton\footnote{An automaton is unambiguous if it has at most one accepting run on every input.} with $n$ states accepting a language~$L$, the probability~$\probmes_s(L)$ can be computed in polynomial time in $\size{\dtmcsymbol}$ and $n$~\cite[Theorem 2]{DBLP:journals/jcss/BaierK00023}.

We begin by considering the always operator and then deal with the remaining operators, as we can reuse the machinery developed for the always operator to deal with them.
By definition, we have $s \in \sat{\prob{\sim \lambda}{\Boxdot\psi'}}{t}$ iff 
\begin{itemize}
    \item $t = 1111$ and $\probmes_s(\safety(\sat{\psi'}{1111})) \sim \lambda$,
    \item $t = 0111$ and $\probmes_s(\cobuchi(\sat{\psi'}{0111})) \sim \lambda$,
    \item $t = 0011$ and $\probmes_s(\buchi(\sat{\psi'}{0011})) \sim \lambda$, and
    \item $t = 0001$ and $\probmes_s(\reach(\sat{\psi'}{0001})) \sim \lambda$,
\end{itemize}
where 
\begin{itemize}
    \item $\safety(S') = \set{\path \in \allpaths{\dtmcsymbol} \mid \state{\path}{n} \in S' \text{ for all } n\in\nats }$, 
    
    \item $\cobuchi( S') = \set{\path \in \allpaths{\dtmcsymbol} \mid \state{\path}{n} \in S' \text{ for all but finitely many } n\in\nats }$,
        
    \item $\buchi(S') = \set{\path \in \allpaths{\dtmcsymbol} \mid \state{\path}{n} \in S' \text{ for infinitely many } n\in\nats }$,
            
     \item $\reach( S') = \set{\path \in \allpaths{\dtmcsymbol} \mid \state{\path}{n} \in S' \text{ for some } n\in\nats }$.
\end{itemize}
All these sets are accepted by some unambiguous Büchi automaton with at most three states (see Figure~\ref{fig_automata}).
As the satisfiability sets~$\sat{\psi'}{t}$ are already computed by induction assumption, we only need to compute $\probmes_s(L)$ for these languages and compare it to the given threshold~$\lambda$.
This can be achieved in polynomial time as argued above.

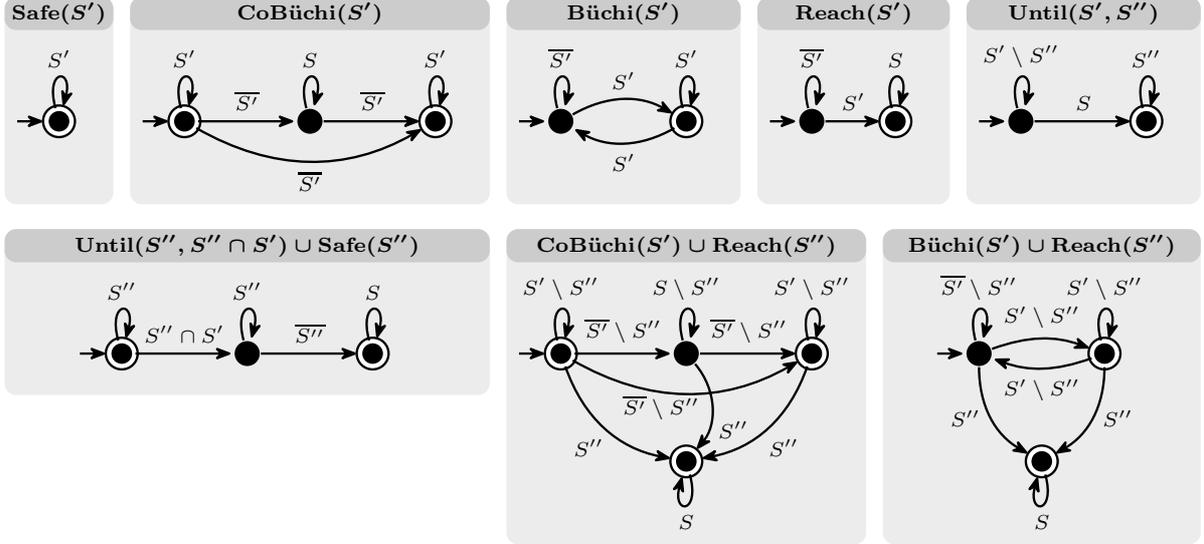
\begin{figure}
    \centering
    \scalebox{1.1}{
\begin{tikzpicture}[thick, shorten >=1pt,node distance=2cm,on grid,>={Stealth[round]},
    st/.style={circle, minimum size = .3cm, inner sep = 0, fill=black},
    ast/.style={circle, double, double distance=1.5pt, minimum size = .3cm, inner sep = 0, draw= black, fill =black}]

\fill[fill=gray!15,rounded corners] (-.65,1.5) rectangle (.65,-1);
\fill[fill=gray!40,rounded corners] (-.65,1.5) rectangle (.65,1.1);
 \node at (0,1.3) {\scriptsize{\boldmath$\safety(S')$}};
  \node[ast]  (a0) at (0,0)                     {};

  \path[->] 
  (-.5,0) edge (a0)
  (a0) edge[loop above]              node []  {\scriptsize$S'$} ()
            ;
            
\fill[fill=gray!15,rounded corners] (.85,1.5) rectangle (5.15,-1);
\fill[fill=gray!40,rounded corners] (.85,1.5) rectangle (5.15,1.1);
\node at (3,1.3) {\scriptsize{\boldmath$\cobuchi(S')$}};
  \node[ast]  (b0) at (1.5,0)                     {};
  \node[st]  (b1) at (3,0)                     {};
  \node[ast]  (b2) at (4.5,0)                     {};

  \path[->] 
  (1,0) edge (b0)
  (b0) edge[loop above]              node []  {\scriptsize$S'$} ()
  (b1) edge[loop above]              node []  {\scriptsize$S$} ()  
  (b2) edge[loop above]              node []  {\scriptsize$S'$} () 
  (b0) edge node[above] {\scriptsize$\overline{S'}$} (b1)
  (b0) edge[bend right] node[below] {\scriptsize$\overline{S'}$} (b2)
  (b1) edge node[above] {\scriptsize$\overline{S'}$} (b2)
;

\fill[fill=gray!15,rounded corners] (5.35,1.5) rectangle (8.15,-1);
\fill[fill=gray!40,rounded corners] (5.35,1.5) rectangle (8.15,1.1);
   \node at (6.75,1.3) {\scriptsize{\boldmath$\buchi(S')$}};
  \node[st]  (c0) at (6,0)                     {};
\node[ast]  (c1) at (7.5,0)                     {};

  \path[->] 
  (5.5,0) edge (c0)
  (c0) edge[loop above]              node []  {\scriptsize$\overline{S'}$} ()
  (c1) edge[loop above]              node []  {\scriptsize$S'$} ()
  (c0) edge[bend left] node[above] {\scriptsize$S'$} (c1)
  (c1) edge[bend left] node[below] {\scriptsize$S'$} (c0)
            ;
\fill[fill=gray!15,rounded corners] (8.35,1.5) rectangle (10.65,-1);
\fill[fill=gray!40,rounded corners] (8.35,1.5) rectangle (10.65,1.1);
\node at (9.5,1.3) {\scriptsize{\boldmath$\reach(S')$}};
  \node[st]  (d0) at (9,0)                     {};
  \node[ast]  (d1) at (10,0)                     {};

  \path[->] 
  (8.5,0) edge (d0)
  (d0) edge[loop above]              node []  {\scriptsize$\overline{S'}$} ()
  (d1) edge[loop above]              node []  {\scriptsize$S$} ()  
  (d0) edge node[above] {\scriptsize$S'$} (d1)
;         

\fill[fill=gray!15,rounded corners] (10.85,1.5) rectangle (13.65,-1);
\fill[fill=gray!40,rounded corners] (10.85,1.5) rectangle (13.65,1.1);
\node at (12.25,1.3) {\scriptsize{\boldmath$\untillang(S',S'')$}};
  \node[st]  (e0) at (11.5,0)                     {};
  \node[ast]  (e1) at (13,0)                     {};

  \path[->] 
  (11,0) edge (e0)
  (e0) edge[loop above]              node []  {\scriptsize$S'\setminus S''$} ()
  (e1) edge[loop above]              node []  {\scriptsize$S''$} ()  
  (e0) edge node[above] {\scriptsize$S$} (e1)
;         

\fill[fill=gray!15,rounded corners] (-.65,-1.3) rectangle (5.15,-3.3);
\fill[fill=gray!40,rounded corners] (-.65,-1.3) rectangle (5.15,-1.7);
\node at (2.25,-1.5) {\scriptsize{\boldmath$\untillang(S'', S''\cap S')\cup \safety(S'')$}};
\node[ast] (f0) at (0.75,-2.8) {};
\node[st] (f1) at (2.25,-2.8) {};
\node[ast] (f2) at (3.75,-2.8) {};

  \path[->] 
  (.25,-2.8) edge (f0)
  (f0) edge[loop above]              node []  {\scriptsize$S''$} ()
  (f1) edge[loop above]              node []  {\scriptsize$S''$} ()
  (f2) edge[loop above]              node []  {\scriptsize$S$} ()
  (f0) edge[] node[above] {\scriptsize$S''\cap S'$} (f1)
  (f1) edge[] node[above] {\scriptsize$\overline{S''}$} (f2)
  ;

\fill[fill=gray!15,rounded corners] (5.35,-1.3) rectangle (9.65,-5.1);
\fill[fill=gray!40,rounded corners] (5.35,-1.3) rectangle (9.65,-1.7);
\node at (7.5,-1.5) {\scriptsize{\boldmath$\cobuchi(S')\cup \reach(S'')$}};
  \node[ast]  (g0) at (6,-2.8)                     {};
  \node[st]  (g1) at (7.5,-2.8)                     {};
  \node[ast]  (g2) at (9,-2.8)                     {};
\node[ast]  (g3) at (7.5,-4.1)                     {};  

  \path[->] 
  (5.5,-2.8) edge (g0)
  (g0) edge[loop above]              node []  {\scriptsize$S' \setminus S''$} ()
  (g1) edge[loop above]              node []  {\scriptsize$S \setminus S''$} ()  
  (g2) edge[loop above]              node []  {\scriptsize$S' \setminus S''$} () 
  (g0) edge node[above] {\scriptsize$\overline{S'} \setminus S''$} (g1)
  (g0) edge[bend right] node[below,xshift=-.3cm,yshift=.1cm] {\scriptsize$\overline{S'} \setminus S''$} (g2)
  (g1) edge node[above] {\scriptsize$\overline{S'} \setminus S''$} (g2)
  (g0) edge[bend right] node[below,xshift=-.2cm] {\scriptsize$S''$} (g3)
  (g1) edge[bend left=40] node[right,near end] {\scriptsize$S''$} (g3)
  (g2) edge[bend left] node[below,xshift=.2cm] {\scriptsize$S''$} (g3)
  (g3) edge[loop below]              node []  {\scriptsize$S$} ()
;

\fill[fill=gray!15,rounded corners] (9.85,-1.3) rectangle (13.65,-5.1);
\fill[fill=gray!40,rounded corners] (9.85,-1.3) rectangle (13.65,-1.7);
\node at (11.75,-1.5) {\scriptsize{\boldmath$\buchi(S')\cup \reach(S'')$}};
    \node[st]  (h0) at (11,-2.8)                     {};
    \node[ast]  (h1) at (12.5,-2.8)                     {};
\node[ast]  (h2) at (11.75,-4.1)                     {};

  \path[->] 
  (10.5,-2.8) edge (h0)
  (h0) edge[loop above]              node []  {\scriptsize$\overline{S'}\setminus S''$} ()
  (h1) edge[loop above]              node []  {\scriptsize$S'\setminus S''$} ()
  (h0) edge[bend left=20] node[above] {\scriptsize$S'\setminus S''$} (h1)
  (h1) edge[bend left=20] node[below] {\scriptsize$S'\setminus S''$} (h0)
  (h0) edge[bend right] node[left] {\scriptsize$S''$} (h2)
  (h1) edge[bend left] node[right] {\scriptsize$S''$} (h2)  
  (h2) edge[loop below]              node []  {\scriptsize$S$} ()
;

  \end{tikzpicture}}

    \caption{The unambiguous Büchi automata for the path properties used in the \rpctl model-checking algorithm. Transitions are labeled by sets of states that represent all states in them (recall that $S$ is the set of all states while $S'$ and $S''$ are subsets of $S$). $\overline{S'}$ denotes the complement of $S'$ w.r.t.\ $S$.}
    \label{fig_automata}
\end{figure}

Now, let us consider the eventually operator. By definition, we have $s \in \sat{\prob{\sim \lambda}{\Diamonddot\psi'}}{t}$ iff $\probmes_s(\reach(\sat{\psi'}{t})) \sim \lambda$, which we have just seen how to check in polynomial time. For the until operator, we have $s \in \sat{\prob{\sim \lambda}{\psi'\Untildot\psi''}}{t}$ iff 
$\probmes_s(\untillang(\sat{\psi'}{t}, \sat{\psi''}{t})) \sim \lambda$, where 
\[
\untillang(S', S'') = \set{\path \in \allpaths{\dtmcsymbol} \mid \text{there is an } n \in \nats \text{ s.t. } \state{\path}{n} \in S'' \text{ and } \state{\path}{n'} \in S' \text{ for all } n'<n }.
\]
There is an unambiguous Büchi automaton with two states accepting this language (see Figure~\ref{fig_automata}). Thus, $\probmes_s(\untillang(\sat{\psi'}{t}, \sat{\psi''}{t})) \sim \lambda$ can again be checked in polynomial time. 

Finally, we consider the release operator. By definition, we have $s \in \sat{\prob{\sim \lambda}{\psi'\Releasedot\psi''}}{t}$ iff 
\begin{itemize}
    \item $t = 1111$ and $\probmes_s([\untillang(\sat{\psi''}{1111},\sat{\psi'}{1111}\cap \sat{\psi''}{1111} )] \cup \safety(\sat{\psi''}{1111})) \sim \lambda$,
    \item $t = 0111$ and $\probmes_s(\cobuchi(\sat{\psi''}{0111} \cup \reach(\sat{\psi'}{0111}))) \sim \lambda$,
    \item $t = 0011$ and $\probmes_s(\buchi(\sat{\psi''}{0011}\cup \reach(\sat{\psi'}{0011})) \sim \lambda$, and
    \item $t = 0001$ and $\probmes_s(\reach(\sat{\psi''}{0001}\cup \reach(\sat{\psi'}{0001})) \sim \lambda$. Note that $\reach(S') \cup \reach(S'') = \reach(S' \cup S'')$ for all sets~$S'$ and $S''$, i.e., we can rely on the results for $\reach$ shown above.
\end{itemize}
Again, all these languages are accepted by unambiguous Büchi automata (see Figure~\ref{fig_automata}) with at most four states, which implies that we can again decide $\probmes_s(L)\sim\lambda$ in polynomial time for these languages~$L$.

Altogether, our algorithm inductively computes $5\size{\cl(\phi)}$ many satisfaction sets, each one in polynomial time (in $\size{\dtmcsymbol}$), and then checks whether $s_\initmark \in \sat{\phi}{t^*}$.
Thus, the algorithm has polynomial running time.
\end{proof}

Again, this result is in line with previous work on robustifying temporal logics: The robustification comes for free (here in terms of computational complexity of the model-checking problem) and the algorithms for the classical semantics can be adapted to handle the robust semantics as well.

\section{Robust PCTL\boldmath$^*$}

In this section, we robustify \pctlstar~\cite{pctlstar}. In line with the general approach, \rpctlstar and \pctlstar share the same syntax (but for the dots), i.e., the formulas of \rpctl are either state formulas or path formulas.
State formulas are given by the grammar
\[
\phi  \Coloneqq p \mid \neg \phi \mid \phi \wedge \phi \mid \phi \vee \phi \mid \phi \rightarrow \phi  \mid \prob{\sim\lambda}{\Phi}
\]
where $p$ ranges over $\ap$, ${\sim} \in \set{<,\le, =, \ge, >}$, $\lambda \in [0,1]$ is a rational probability threshold, and $\Phi$ ranges over path formulas.
Path formulas are given~by 
\[
\Phi \Coloneqq \phi \mid \neg \Phi \mid \Phi \wedge \Phi \mid \Phi \vee \Phi \mid \Phi \rightarrow \Phi  \mid \Nextdot \Phi \mid \Diamonddot \Phi \mid \Boxdot \Phi \mid \Phi \Untildot \Phi \mid \Phi \Releasedot \Phi
\]
where $\phi$ ranges over state formulas.
Formula size is defined as for \rpctl.

Also, \rpctlstar is evaluated over discrete-time Markov chains, just as \pctlstar. Let \dtmc~$\dtmcsymbol = (S, s_\initmark, \delta, \ell)$ be a \dtmc. 
The semantics of \rpctlstar is again defined via an evaluation function~$\reval_\dtmcsymbol$, this time mapping a vertex~$s$ of $\dtmcsymbol$ and a state formula, or a path of $\dtmcsymbol$ and a path formula to a truth value in $\boolfour$.

As \rpctlstar is designed to extend \rpctl, the definition of the \rpctlstar semantics is (in some parts) very similar to that of \rpctl. This is in particular true for the Boolean connectives (both for state and path formulas), which is exactly the same as for \rpctl. However, to easily accommodate the arbitrary nesting of temporal operators in \rpctlstar, we use an alternative definition of the semantics for path formulas, which mimics the original semantics for \rltl~\cite{TabuadaNeider}. This follows the precedent of robust \ctlstar~\cite{rctljournal}, where the semantics of path formulas is derived from the semantics of \rltl as well.

The \rpctlstar evaluation function is defined inductively via
\begin{itemize}
    
    \item $\reval_\dtmcsymbol(s, p) = \begin{cases} 
    1111 & \text{if } p \in \ell(s),\\
    0000 & \text{if } p \notin \ell(s),
    \end{cases}
    $
    
    \item $\reval_\dtmcsymbol(s, \neg \phi) = \begin{cases} 
    1111 & \text{if } \reval_\dtmcsymbol(s, \phi) \prec 1111,\\
    0000 & \text{if } \reval_\dtmcsymbol(s,\phi) = 1111,
    \end{cases}
    $
    
    \item $\reval_\dtmcsymbol(s, \phi_0 \wedge \phi_1) = \min(\reval_\dtmcsymbol(s, \phi_0), \reval_\dtmcsymbol(s, \phi_1))$,
    
    \item $\reval_\dtmcsymbol(s, \phi_0 \vee \phi_1) = \max(\reval_\dtmcsymbol(s, \phi_0), \reval_\dtmcsymbol(s, \phi_1))$,
    
    \item $\reval_\dtmcsymbol(s, \phi_0 \rightarrow \phi_1) = \begin{cases} 
    1111 & \text{if } \reval_\dtmcsymbol(s, \phi_0) \preceq \reval_\dtmcsymbol(s, \phi_1),\\
    \reval_\dtmcsymbol(s, \phi_1) & \text{if } \reval_\dtmcsymbol(s, \phi_0) \succ \reval_\dtmcsymbol(s, \phi_1),
    \end{cases}
    $
    
    \item $\reval_\dtmcsymbol(s, \prob{\sim\lambda}{\Phi}) = \max \set{t \in \boolfour \mid \probmes_s(\set{\pi \in\pointedpaths{\dtmcsymbol}{s} \mid \reval_\dtmcsymbol(\pi, \Phi) \succeq t}) \sim \lambda}$ with the convention~$\max \emptyset = 0000$,

    \item $\reval_\dtmcsymbol(\pi, \phi) = \reval_\dtmcsymbol(\state{\pi}{0},\phi)$,
    
        \item $\reval_\dtmcsymbol(\pi, \neg \Phi) = \begin{cases} 
    1111 & \text{if } \reval_\dtmcsymbol(\pi, \Phi) \prec 1111,\\
    0000 & \text{if } \reval_\dtmcsymbol(\pi,\Phi) = 1111,
    \end{cases}
    $

    \item $\reval_\dtmcsymbol(\pi, \Phi_0 \wedge \Phi_1) = \min(\reval_\dtmcsymbol(\pi, \Phi_0), \reval_\dtmcsymbol(\pi, \Phi_1))$,
    
    \item $\reval_\dtmcsymbol(\pi, \Phi_0 \vee \Phi_1) = \max(\reval_\dtmcsymbol(\pi, \Phi_0), \reval_\dtmcsymbol(\pi, \Phi_1))$,
    
    \item $\reval_\dtmcsymbol(\pi, \Phi_0 \rightarrow \Phi_1) = \begin{cases} 
    1111 & \text{if } \reval_\dtmcsymbol(\pi, \Phi_0) \preceq \reval_\dtmcsymbol(\pi, \Phi_1),\\
    \reval_\dtmcsymbol(s, \Phi_1) & \text{if } \reval_\dtmcsymbol(\pi, \Phi_0) \succ \reval_\dtmcsymbol(\pi, \Phi_1),
    \end{cases}
    $
    
    \item $\reval_\dtmcsymbol(\pi, \Nextdot\Phi) = \reval_\dtmcsymbol(\suffix{\pi}{1},\Phi)$,
    
    \item $\reval_\dtmcsymbol(\pi, \Diamonddot\Phi) = b_1 b_2 b_3 b_4$ with $b_k = \max_{n \ge 0} (\reval_\dtmcsymbol(\suffix{\pi}{n},\Phi))[k]$ for all $k\in\set{1,2,3,4}$, 
    
    \item $\reval_\dtmcsymbol(\pi, \Boxdot \Phi) = b_1 b_2 b_3 b_4$ with
    \begin{itemize}
        \item $b_1 = \min_{n \ge 0}( \reval_\dtmcsymbol(\suffix{\pi}{n},\Phi))[1]$,
        \item $b_2 = \max_{m \ge 0}( \min_{n \ge m} \reval_\dtmcsymbol(\suffix{\pi}{n},\Phi))[2]$,
        \item $b_3 = \min_{m \ge 0}( \max_{n \ge m} \reval_\dtmcsymbol(\suffix{\pi}{n},\Phi))[3]$, and
        \item $b_4 = \max_{n \ge 0}( \reval_\dtmcsymbol(\suffix{\pi}{n},\Phi))[4]$,
    \end{itemize}
    
\item $\reval_\dtmcsymbol(\pi, \Phi\Untildot \Psi) = b_1 b_2 b_3 b_4$ with

$b_k = \max_{n \ge 0} \min\set{
(\reval_\dtmcsymbol(\suffix{\pi}{n},\Psi))[k], \min\set{ (\reval_\dtmcsymbol(\suffix{\pi}{n'},\Phi))[k] \mid 0 \le n' < n}
} $ for all $k\in\set{1,2,3,4}$, and

\item $\reval_\dtmcsymbol(\pi, \Phi\Releasedot \Psi) = b_1 b_2 b_3 b_4$ with 
\begin{itemize}
    \item $b_1 = \min_{n' \ge 0} \max \set{ (\reval_\dtmcsymbol(\suffix{\pi}{n'}, \Psi))[1]  , \max_{n'' < n'}(\reval_\dtmcsymbol(\suffix{\pi}{n''}, \Phi))[1] }$,
    \item $b_2 = \max_{n \ge 0} \min_{n' \ge n} \max \set{ (\reval_\dtmcsymbol(\suffix{\pi}{n'}, \Psi))[2]  , \max_{n'' < n'}(\reval_\dtmcsymbol(\suffix{\pi}{n''}, \Phi))[2] }$,
    \item $b_3 = \min_{n \ge 0} \max_{n' \ge n} \max \set{ (\reval_\dtmcsymbol(\suffix{\pi}{n'}, \Psi))[3]  , \max_{n'' < n'}(\reval_\dtmcsymbol(\suffix{\pi}{n''}, \Phi))[3] }$, and
    \item $b_4 = \max_{n' \ge 0} \max \set{ (\reval_\dtmcsymbol(\suffix{\pi}{n'}, \Psi))[4]  , \max_{n'' < n'}(\reval_\dtmcsymbol(\suffix{\pi}{n''}, \Phi))[4] }$.
\end{itemize}
\end{itemize}

Note that while the definition of the semantics of the temporal operators differs from the one for \rpctl (to easily accomodate arbitrary nesting of temporal operators which is not possible in \rpctl), \rpctl is a fragment of \rpctlstar.

\begin{example}
Consider the formula~$\prob{\ge .9}{\Boxdot a \rightarrow \Boxdot g}$, a variant of the assume-guarantee property of Example~\ref{example:rpctl}.
It evaluates to the largest truth value~$t$ such that $\Boxdot a \rightarrow \Boxdot g$ evaluates to $t$ with probability $\ge .9$.
Now, on a single path, $\Boxdot a \rightarrow \Boxdot g$ evaluates to
\begin{itemize}
    \item $1111$ if $\Boxdot g$ evaluates to a larger or equal truth value than $\Boxdot a$ and 
    \item to $t \prec 1111$ if $\Boxdot a$ evaluates to $t$ and $\Boxdot g$ evaluates to a truth value larger than $t$.
\end{itemize}
\end{example}

\subsection{Expressiveness}

As usual for temporal logics that allow arbitrary nesting of temporal operators (e.g., \ltl, \ctlstar, and \atlstar) \rpctlstar has the same expressiveness as its non-robust version: the first bit of the five-valued semantics of \rpctlstar again captures the semantics of non-robust \pctlstar (as per design goals), thereby yielding the first embedding, while arbitrary nesting of temporal operators allows to mimic the five-valued semantics of \rpctlstar explicitly in non-robust \pctlstar, there yielding the second embedding.

\begin{theorem}
\label{thm:rpctlstarexpressiveness}
\rpctlstar is as expressive as \pctlstar. Both translations can be computed in polynomial time.
\end{theorem}

\begin{proof}
The translation from \pctlstar to \rpctlstar is a generalization of the analogous result for \pctl and \rpctl (see Lemma~\ref{lemma:pctl2rpctl}): 
Let $\phi$ be a \pctlstar state formula without implications, and let $\dtmcsymbol$ be a \dtmc with initial state~$s_\initmark$. Then, $\dtmcsymbol,s_\initmark \models \phi$ iff $\reval_\dtmcsymbol(s_\initmark, \dot{\phi}) = 1111$, where $\dot{\phi}$ is the \rpctlstar state formula obtained from $\phi$ by dotting all temporal operators. 
This  is again proven by induction over the construction of $\phi$.

For the other direction, we inductively translate an \rpctlstar state formula~$\phi$ and a truth value~$t \in\boolfour$ into a \pctlstar state formula~$\phi_t$ such that $\reval_\dtmcsymbol(s,\phi) \succeq t$ iff $\dtmcsymbol,s\models \phi_t$ for all {\dtmc}s~$\dtmcsymbol$ and all states $s$ of $\dtmcsymbol$.
This is in line with previous work on \ltl~\cite{TabuadaNeider}, \ctlstar~\cite{rctljournal}, and \atlstar~\cite{MNZ23}.

As we have $\reval_\dtmcsymbol(s,\phi) \succeq 0000$ for all state formulas~$\phi$, we define $\phi_{0000}$ to be some tautology (say $p \vee \neg p$) and only consider $t \succ 0000$ in the following. We start with atomic propositions and define $p_t = p$.
The translation for Boolean connectives is the same for state and path formulas.
So, to avoid duplication, $\chi$ ranges in the following over state and path formulas.
\begin{itemize}    
    \item $(\neg \chi)_t = \neg \chi_t$,
    \item $(\chi_1 \vee \chi_2)_t = (\chi_1)_t \vee (\chi_2)_t$ and  $(\chi_1 \wedge \chi_2)_t = (\chi_1)_t \wedge (\chi_2)_t$,
    \item $(\chi_1 \rightarrow \chi_2)_{1111} = \bigwedge_{t \succeq 0000} (\chi_2)_t \vee \neg (\chi_1)_t$, and 
    \item $(\chi_1 \rightarrow \chi_2)_{t} = (\chi_1 \rightarrow \chi_2)_{1111} \vee (\chi_2)_t$ for $t \prec 1111$.
\end{itemize}
Next, we define $(\prob{\sim\lambda}{\Phi})_t = \prob{\sim\lambda}{\Phi_t}$ and consider the temporal operators:
\begin{itemize}
    \item $(\Nextdot\Phi)_t = \Next\Phi_t$ and $(\Diamonddot\Phi)_t = \Diamond\Phi_t$,
    \item $(\Boxdot\Phi)_{1111} = \Box\Phi_{1111}$,  $(\Boxdot\Phi)_{0111} = \Diamond\Box\Phi_{0111}$,  $(\Boxdot\Phi)_{0011} = \Box\Diamond\Phi_{0011}$, and $(\Boxdot\Phi)_{0001} = \Diamond\Phi_{0001}$,
    \item $(\Phi\Untildot\Psi)_t = \Phi_t\Until\Psi_t$,
    \item $(\Phi \Releasedot \Psi)_{1111} = \Phi_{1111} \Releasedot \Psi_{1111}$,
     $(\Phi \Releasedot \Psi)_{0111} = \Diamond\Box\Psi_{0111} \vee \Diamond \Phi_{1111}$,
    and $(\Phi \Releasedot \Psi)_{0011} = \Box\Diamond\Psi_{0011} \vee \Diamond \Phi_{1111}$, and
     $(\Phi \Releasedot \Psi)_{0001} = \Diamond\Psi_{0001} \vee \Diamond \Phi_{1111}$.
\end{itemize}
An induction over the construction of $\phi$ shows that $\phi_t$ has the desired properties.
\end{proof}

\subsection{Model-checking}

The model-checking problem for \rpctlstar is defined as for \rpctl: Given a \dtmc~$\dtmcsymbol$ with initial state~$s_\initmark$, an \rpctlstar state formula~$\phi$, and a truth value~$t^* \in \boolfour$, is $\reval_\dtmcsymbol(s, \phi)\succeq t^*$?
It is \pspace-complete, as is the \pctlstar model-checking problem~\cite{pctlstar,VardiWolper86}, i.e., robustness comes again for free.
This results follows directly from the fact that \rpctlstar can be (in polynomial time) translated into \pctlstar and follows previous results on robust \ctlstar~\cite{rctljournal} and robust \atlstar~\cite{MNZ23}.

\begin{theorem}
\rpctlstar model-checking is \pspace-complete.
\end{theorem}

\begin{proof}
The result follows immediately from Theorem~\ref{thm:rpctlstarexpressiveness} and the \pspace-completeness of \pctlstar model-checking.
\end{proof}

\section{Related Work}

There is a plethora of work on the verification of probabilistic systems and on robustifying verification. 
Due to space restrictions, we focus here on the intersection of these two areas, which is our concern in this work.

A major challenge in the modelling of probabilistic systems is the fact that determining exact transition probabilities is often impossible. Instead one resorts to statistical analyses of the system, which comes with uncertainties.\footnote{In fact, even the work introducing \pctlstar considered models with unknown transition probabilities~\cite{pctlstar}.} 
However, verification results are often highly sensitive to changes in the transition probabilities, i.e., modelling and verification are not robust to those changes.
Hence, a large body of work is concerned with capturing uncertainty in probabilistic systems and their subsequent verification. 

Various types of uncertain transition functions for Markov chains have been introduced, e.g., interval bounded {\dtmc}s~\cite{senetal} where only upper and lower bounds on the transition probabilities are specified, and convex MDPs, Markov decision processes with convex uncertainties~\cite{puggellietal}, and robust MDPs with rectangular ambiguity sets~\cite{Iyengar,Nilimetal}.
However, there are uncertainties beyond the transition probabilities, e.g., in the form of partial observability and adversarial behaviour. A recent position paper by Badings et al.~\cite{badingsetal} gives a thorough overview of the state-of-the-art in decision making under uncertainty, presenting a survey of uncertainty models that enable more robust modelling and verification.
Finally, other approaches to handling uncertainty include simulation~\cite{ashoketal,wiesemannetal} and approximation~\cite{jaegeretal}.

\section{Conclusion}

We have shown how to robustify \pctl and \pctlstar, obtaining the logics \rpctl and \rpctlstar. 
The model-checking problems for these robust logics are as hard as the model-checking problems for the non-robust variants, i.e., robustness can be added for free. 
This is in line with previous work on robust variants of \ltl~\cite{DBLP:journals/tocl/AnevlavisPNT22} and its extensions~\cite{DBLP:journals/iandc/NeiderWZ22}, as well as \ctl and \ctlstar~\cite{rctljournal}, and \atl and \atlstar~\cite{MNZ23}. 

Probably the most interesting problem left for future work concerns the expressiveness of \rpctl and \pctl.
Note that in the non-probabilistic setting, it is known that \rctl is strictly more expressive than \ctl~\cite[Section 3.3]{rctljournal}.
However, as discussed in Subsection~\ref{subsec_exp}, it is unclear whether this separation can be lifted to the probabilistic setting.

\paragraph{Acknowledgements} We want to thank Marco Muñiz for proposing to study \rpctl and \rpctlstar and for many fruitful discussions, as well as the reviewers for their valuable feedback. 

This work was supported by DIREC – Digital Research Centre Denmark.

\bibliographystyle{plain}

\bibliography{biblio.bib}

\end{document}